\documentclass[11pt]{article}
\usepackage{odonnell}

\newcommand{\decode}{\mathrm{Dec}}
\renewcommand{\even}{\mathrm{EVEN}}

\newcommand{\threesat}{\textsf{3Sat}\xspace}
\newcommand{\alin}[1]{\textsf{#1NLin}\xspace}
\newcommand{\notall}{\textsf{4-Not-All-There}\xspace}
\newcommand{\coloring}{\textsf{3-Coloring}\xspace}
\newcommand{\kcoloring}{k\textsf{-Coloring}\xspace}
\newcommand{\fournat}{\textsf{4NAT}\xspace}
\newcommand{\csp}{\mathsf{CSP}}
\newcommand{\instance}{\mathcal{I}}
\newcommand{\buddy}{\textsf{TwoPair}\xspace}
\newcommand{\twopair}{\buddy}
\newcommand{\gadg}{\gamma}
\newcommand{\roots}{U_3}
\newcommand{\x}{\bx}
\newcommand{\y}{\by}
\newcommand{\z}{\bz}
\newcommand{\w}{\bw}
\newcommand{\pih}{\pi_3}

\begin{document}

\title{New $\NP$-hardness results for 3-Coloring and 2-to-1 Label Cover\footnote{A preliminary version of these results appeared as~\cite{AOW12}}}

\author{Per Austrin\thanks{Department of Computer Science, University of Toronto.  Funded by NSERC.} \and Ryan O'Donnell\thanks{Department of Computer Science, Carnegie Mellon University. Supported by NSF grants CCF-0747250 and CCF-0915893, and by a Sloan fellowship.} \and Li-Yang Tan\thanks{Department of Computer Science, Columbia University. Research done while visiting CMU.} \and John Wright\thanks{Department of Computer Science, Carnegie Mellon University.}}

\maketitle

\begin{abstract}
We show that given a 3-colorable graph, it is $\NP$-hard to find a 3-coloring with $(\frac{16}{17}+\eps)$ of the edges bichromatic.  In a related result, we show that given a satisfiable instance of the $2$-to-$1$ Label Cover problem, it is $\NP$-hard to find a $(\frac{23}{24} + \eps)$-satisfying assignment. 
\end{abstract}

\setcounter{page}{0}
\thispagestyle{empty}
\newpage

\section{Introduction}

Graph coloring problems differ from many other Constraint Satisfaction Problems (CSPs) in that we typically care about the case of \emph{perfect completeness}, e.g.~when the graph under consideration is 3-colorable rather than almost 3-colorable.  Unfortunately, this means that many of the powerful tools which have been developed for proving inapproximability results are no longer applicable.  Most prominently, Raghavendra's~\cite{Rag08} optimal inapproximability results for all CSPs, which are conditioned on the unproven \emph{Unique Games Conjecture} (UGC), only apply to the case of imperfect completeness.  The UGC states that it is NP-hard to distinguish between nearly satisfiable and almost completely unsatisfiable instances of \emph{Unique}, \emph{or 1-to-1}, Label Cover.  As a result, by starting a reduction with the UGC, one has already lost perfect completeness.  Thus, any inapproximability result for a graph coloring problem must begin with a different unproven assumption, such as $\PTIME \neq \NP$ or Khot's~\cite{Kho02} 2-to-1 Conjecture.

The motivation of this paper is to study both of these assumptions as they relate to the graph $\kcoloring$ problem, specifically in the $k = 3$ case.  In the $\kcoloring$ problem, the input is a $k$-colorable graph $G$, and the task is to find a $k$-coloring of the vertices of $G$ which maximizes the number of bichromatic edges.  This problem has previously gone under the names ``Max-$k$-Colorability''~\cite{Pet94} and ``Maximum $k$-Colorable Subgraph''~\cite{GS09}.  Graph $\kcoloring$, along with its many studied variants, is a central problem in Computer Science, and pinning down its exact approximability is an important open problem.  The main result of our paper is an improved inapproximability result for $\coloring$, predicated only on $\PTIME \neq \NP$:
\begin{theorem}\label{thm:3-coloring-hardness}
For all $\eps > 0$, $(1, \frac{16}{17} + \eps)$-deciding the $\coloring$ problem is $\NP$-hard.
\end{theorem}
Here by $(c, s)$-deciding a CSP we mean the task of determining whether an instance is at least $c$-satisfiable or less than $s$-satisfiable.  In fact, this  is the best known hardness result for the $\coloring$ problem, even assuming conjectures such as the $2$-to-$1$ Conjecture.  The previous best $\NP$-hardness for $\coloring$ was due to Guruswami and Sinop~\cite{GS09}, who showed a factor $\frac{32}{33}$-hardness via a somewhat involved gadget reduction from the 3-query adaptive PCP result of~\cite{GLST98}.  In contrast, the best current algorithm achieves an approximation ratio of~$0.836$ (and does not need the instance to be satisfiable)~\cite{GW04}.  As for larger values of $k$, \cite{GS09} construct a reduction which directly translates hardness results for $\coloring$ into hardness results for $\kcoloring$, for $k \geq 3$.  Applying this to our Theorem~\ref{thm:3-coloring-hardness} yields
\begin{theorem}
For all $k \geq 3$ and $\eps > 0$, it is $\NP$-hard to $(1, 1- \frac{1}{17(k + c_k) + c_k} +\eps)$-decide the $\kcoloring$ problem.  Here $c_k = k \pmod{3}$.
\end{theorem}

This is the best known $\NP$-hardness for $\kcoloring$.  For sufficiently large $k$, stronger inapproximability results are known to follow from the 2-to-1 Conjecture:
\begin{named}{2-to-1 Conjecture}[\cite{Kho02}]
For every integer $\eps > 0$, there is a label set size~$q$ such that it is $\NP$-hard to $(1, \eps)$-decide the $2$-to-$1$ Label Cover problem.
\end{named}
\noindent Guruswami and Sinop~\cite{GS09} have shown that the $2$-to-$1$ Conjecture implies it is $\NP$-hard to $(1, 1 - \frac{1}{k} + O(\frac{\ln k}{k^2}))$-decide the $\kcoloring$ problem.  This result would be tight up to the $O(\cdot)$ by an algorithm of Frieze and Jerrum~\cite{FJ97}.  In a prior result,  Dinur, Mossel, and Regev~\cite{DMR09}  showed that the $2$-to-$1$ Conjecture implies that it is $\NP$-hard to $C$-color a $4$-colorable graph for any constant~$C$. (They also showed hardness for $3$-colorable graphs via another Unique Games variant.) It is therefore clear that settling the $2$-to-1 Conjecture is important to the study of the inapproximability of graph coloring problems.

Interestingly, to a certain extent the reverse is also true: it is ``folklore'' that hardness results for graph $\coloring$ immediately imply hardness results for the 2-to-1 Label Cover problem with label sizes $3$ \& $6$ by a simple ``constraint-variable'' reduction.  Indeed, Theorem~\ref{thm:3-coloring-hardness} by itself would give the best-known $\NP$-hardness for 2-to-1 Label Cover.  However, we are able to get an even better hardness result than this by studying a CSP closely related to $\coloring$.  Our hardness result is:
\begin{theorem}\label{thm:twotoone-hardness}
For all $\eps > 0$, $(1, \frac{23}{24} + \eps)$-deciding the $2$-to-$1$ Label Cover problem with label set sizes $3$ \textnormal{\&} $6$ is $\NP$-hard.  
\end{theorem}
By duplicating labels, this result also holds for label set sizes $3k$ \& $6k$ for any $k \in \N^+$.  To the best of our knowledge, no explicit $\NP$-hardness for this problem has previously been stated in the literature.  Combining the constraint-variable reduction with the above-mentioned $\coloring$ hardness of~\cite{GS09} gives an $\NP$-hardness of $(1, \frac{65}{66} + \eps)$ for the problem with label sizes $3$ \& $6$, which we believe to be the best previously known.  It is not known how to take advantage of larger label set sizes.  On the other hand, for label set sizes $2$~\&~$4$ it is known that satisfying $2$-to-$1$ Label Cover instances can be found in polynomial time.

Regarding the hardness of the $2$-to-$1$ Label Cover problem, the only evidence we have is a family of integrality gaps for the canonical SDP relaxation of the problem, in~\cite{GKO+10}.  Regarding algorithms for the problem, an important recent line of work beginning in~\cite{ABS10} (see also~\cite{BRS11,GS11,Ste10}) has sought subexponential-time algorithms for Unique Label Cover and related problems.  In particular, Steurer~\cite{Ste10} has shown that for any constant $\beta > 0$ and label set size, there is an $\exp(O(n^{\beta}))$-time algorithm which, given a satisfiable $2$-to-$1$ Label Cover instance, finds an assignment satisfying an $\exp(-O(1/\beta^2))$-fraction of the constraints.  E.g., there is a $2^{O(n^{.001})}$-time algorithm which $(1,s_0)$-approximates $2$-to-$1$ Label Cover, where $s_0 > 0$ is a certain universal constant.

In light of this, it is interesting not only to seek $\NP$-hardness results for certain approximation thresholds, but to additionally seek evidence that \emph{nearly full exponential time} is required for these thresholds.  This can done by assuming the Exponential Time Hypothesis (ETH)~\cite{IP01} and by reducing from the Moshkovitz--Raz Theorem~\cite{MR10}, which shows a near-linear size reduction from \threesat to the standard Label Cover problem with subconstant soundness.  In this work, we show reductions from $\threesat$  to the problem of $(1, s+\eps)$-approximating several CSPs, for certain values of $s$ and for all $\eps >0$.  In fact, though we omit it in our theorem statements, it can be checked that all of the reductions in this paper are quasilinear in size for $\eps = \eps(n) = \Theta\left(\frac{1}{(\log \log n)^\beta}\right)$, for some $\beta > 0$.

\subsection{Our techniques} \label{sec:techniques}

Let us describe the high-level idea behind our result.  The folklore constraint-variable reduction from $\coloring$ to $2$-to-$1$ Label Cover would work just as well if we started from ``$\coloring$ with literals'' instead.  By this we mean the CSP with domain $\Z_3$ and constraints of the form ``$v_i - v_j \neq c \pmod{3}$''.  Starting from this CSP --- which we call $\alin{2}(\Z_3)$ --- has two benefits: first, it is at least as hard as $\coloring$ and hence could yield a stronger hardness result; second, it is a bit more ``symmetrical'' for the purposes of designing reductions.  Finally, having proven a hardness result for $\alin{2}$, it seems reasonable that it can be modified into a hardness result for $\coloring$.
We obtain the following hardness result for $\alin{2}(\Z_3)$.
\begin{theorem}\label{thm:alin-hardness}
For all $\eps > 0$, it is $\NP$-hard to $(1, \frac{11}{12} + \eps)$-decide the $\alin{2}$ problem.
\end{theorem}
\noindent
As $\coloring$ is a special case of $\alin{2}(\Z_3)$, $\cite{GS09}$ also shows that $(1, \frac{32}{33}+\eps)$-deciding $\alin{2}$ is $\NP$-hard for all $\eps > 0$, and to our knowledge this was previously the only hardness known  for $\alin{2}(\Z_3)$. Further, the $0.836$-approximation algorithm for $\coloring$ from above achieves the same approximation ratio for $\alin{2}(\Z_3)$, and this is the best known~\cite{GW04}.  To prove Theorem~\ref{thm:alin-hardness}, we proceed by designing an appropriate ``function-in-the-middle'' dictator test, as in the recent framework of~\cite{OW12}.  Although the~\cite{OW12} framework gives a direct translation of certain types of function-in-the-middle tests into hardness results, we cannot employ it in a black-box fashion.  Among other reasons,~\cite{OW12} assumes that the test has ``built-in noise'', but we cannot afford this as we need our test to have perfect completeness.

Thus, we need a different proof to derive a hardness result from this function-in-the-middle test.  We first were able to accomplish this by an analysis similar to the Fourier-based proof of $2\mathsf{Lin}(\Z_2)$ hardness given in Appendix~F of~\cite{OW12}.  Just as that proof ``reveals'' that the function-in-the-middle $2\mathsf{Lin}(\Z_2)$ test can be equivalently thought of as \Hastad's $3\mathsf{Lin}(\Z_2)$ test composed with the $3\mathsf{Lin}(\Z_2)$-to-$2\mathsf{Lin}(\Z_2)$ gadget of~\cite{TSSW00}, our proof for the $\alin{2}(\Z_3)$ function-in-the-middle test revealed it to be the composition of a function test for a certain four-variable CSP with a gadget.  We have called the particular four-variable CSP \notall, or $\fournat$ for short.  Because it is a $4$-CSP, we are able to prove the following $\NP$-hardness of approximation result for it using a classic, \Hastad-style Fourier-analytic proof.
\begin{theorem}\label{thm:fournat-hardness}
For all $\eps > 0$, it is $\NP$-hard to $(1, \frac{2}{3} + \eps)$-decide the $\fournat$ problem.
\end{theorem}
\noindent
Thus, the final form in which we present our Theorem~\ref{thm:twotoone-hardness} is as a reduction from Label-Cover to $\fournat$ using a function test (yielding Theorem~\ref{thm:fournat-hardness}), followed by a $\fournat$-to-$\alin{2}(\Z_3)$ gadget (yielding Theorem~\ref{thm:alin-hardness}), followed by the constraint-variable reduction to $2$-to-$1$ Label Cover.  Indeed, all of the technology needed to carry out this proof was in place for over a decade, but without the function-in-the-middle framework of~\cite{OW12} it seems that pinpointing the $\fournat$ predicate as a good starting point would have been unlikely.

Our proof of Theorem~\ref{thm:3-coloring-hardness} is similar: we design a function-in-the-middle test for $\coloring$ which uses the $\alin{2}(\Z_3)$ function test as a subroutine.  And though we do not find a gadget reduction from $\coloring$ to $\fournat$, we are able to express the success probability of the test in terms of the $\fournat$ test.  Thus, there is significant overlap in the proofs of our two main theorems, and we are able to carry out the proofs simultaneously.

\subsection{Organization}

We leave to Section~\ref{sec:prelims} most of the definitions, including those of the CSPs we use.  The heart of the paper is in Section~\ref{sec:better}, where we give the $\alin{2}(\Z_3)$, $\coloring$, and $\fournat$ function tests and explain how they are related.  Section~\ref{sec:analysis} contains the Fourier analysis of the tests.  The actual hardness proof for $\fournat$ is presented in Section~\ref{app:fournat}, and it follows mostly the techniques put in place by \Hastad in \cite{Has01}.  Because the hardness proof for $\coloring$ is almost identical, we omit it.    Appendix~\ref{sec:eggg} contains a technical lemma.

\section{Preliminaries}\label{sec:prelims}

We primarily work with strings $x \in \Z_3^K$ for some integer $K$.  We write $x_i$ to denote the $i$th coordinate of~$x$. 

\subsection{Definitions of problems}

An instance $\instance$ of a \emph{constraint satisfaction problem} (CSP) is a set of variables $V$, a set of labels $D$, and a weighted list of constraints on these variables.  We assume that the weights of the constraints are nonegative and sum to 1.  The weights therefore induce a probability distribution on the constraints.  Given an assignment to the variables $f:V\rightarrow D$, the \emph{value} of $f$ is the probability that $f$ satisfies a constraint drawn from this probability distribution.  The \emph{optimum} of $\instance$ is the highest value of any assignment.  We say that an $\instance$ is $s$-\emph{satisfiable} if its optimum is at least $s$.  If it is 1-satisfiable we simply call it satisfiable.

We define a CSP $\calP$ to be a set of CSP instances.  Typically, these instances will have similar constraints.  We will study the problem of \emph{$(c, s)$-deciding} $\calP$.  This is the problem of determining whether an instance of $\calP$ is at least $c$-satisfiable or less than $s$-satisfiable.  Related is the problem of \emph{$(c, s)$-approximating} $\calP$, in which one is given a $c$-satisfiable instance of $\calP$ and asked to find an assignment of value at least $s$.  It is easy to see that $(c, s)$-deciding $\calP$ is at least as easy as $(c, s)$-approximating $\calP$.  Thus, as all our hardness results are for $(c, s)$-deciding CSPs, we also prove hardness for $(c, s)$-approximating these CSPs.

We now state the four CSPs that are the focus of our paper.

\paragraph{\textsf{3-coloring}:} In this CSP the label set is $\Z_3$ and the constraints are of the form $v_i \neq v_j$.

\paragraph{\textsf{2-NLin($\Z_3$)}:} In this CSP the label set is $\Z_3$ and the constraints are of the form
\begin{equation*}
v_i - v_j \neq a \pmod{3}, \quad a \in \Z_3.
\end{equation*}
The special case when each RHS is~$0$ is the $\coloring$ problem.  We often drop the $(\Z_3)$ from this notation and simply write $\alin{2}$. The reader may think of the `\textsf{N}' in $\alin{2}(\Z_3)$ as standing for `N'on-linear, although we prefer to think of it as standing for `N'early-linear.  The reason is that when generalizing to moduli $q > 3$, the techniques in this paper generalize to constraints of the form ``$v_i  - v_j \pmod{q} \in \{a, a+1\}$'' rather than ``$v_i - v_j \neq a \pmod{q}$''.  For the ternary version of this constraint, ``$v_i - v_j + v_k \pmod{q} \in \{a,a+1\}$'', it is folklore\footnote{Venkatesan Guruswami, Subhash Khot personal communications.} that a simple modification of \Hastad's work~\cite{Has01} yields $\NP$-hardness of $(1,\frac{2}{q})$-approximation.

\paragraph{\notall:}  For the \notall problem, denoted $\fournat$, we define $\fournat \co \Z_3^4 \rightarrow \{0, 1\}$ to have output $1$ if and only if at least one of the elements of $\Z_3$ is not present among the four inputs.  The $\fournat$ CSP has label set $D = \Z_3$ and constraints of the form $\fournat(v_1 + k_1, v_2 + k_2, v_3 + k_3, v_4 + k_4) = 1$, where the $k_i$'s are constants in $\Z_3$.

We additionally define the ``\textsf{Two Pairs}'' predicate $\buddy \co \Z_3^4 \rightarrow \{0, 1\}$, which has output $1$ if and only if its input contains two distinct elements of $\Z_3$, each appearing twice.  Note that an input which satisfies $\buddy$ also satisfies $\fournat$.

\paragraph{$\mathbf{d}$-to-1 Label Cover:}  An instance of the $d$-to-1 Label Cover problem is a bipartite graph $G=(U \cup V, E)$, a label set size $K$, and a $d$-to-1 map $\pi_e:[dK] \rightarrow [K]$ for each edge $e \in E$.  The elements of $U$ are labeled from the set $[K]$, and the elements of $V$ are labeled from the set $[dK]$.  A labeling $f: U \cup V \rightarrow [dK]$ satisfies an edge $e = (u, v)$ if $\pi_e(f(v)) = f(u)$.  Of particular interest is the $d = 2$ case, i.e., 2-to-1 Label Cover.

Label Cover serves as the starting point for most $\NP$-hardness of approximation results.  We use the following theorem of Moshkovitz and Raz:
\begin{theorem}[\cite{MR10}]\label{thm:mosraz}
For any $\eps = \eps(n) \geq n^{-o(1)}$ there exists $K, d \leq 2^{\mathrm{poly}(1/\eps)}$ such that the problem of deciding a \threesat instance of size $n$ can be Karp-reduced in $\mathrm{poly}(n)$ time to the problem of $(1, \eps)$-deciding  $d$-to-1 Label Cover instance of size $n^{1+o(1)}$ with label set size~$K$.
\end{theorem}

\subsection{Gadgets}

A typical way of relating two separate CSPs is by constructing a \emph{gadget reduction} which translates from one to the other.  A gadget reduction from $\csp_1$ to $\csp_2$ is one which maps any $\csp_1$ constraint into a weighted set of $\csp_2$ constraints.  The $\csp_2$ constraints are over the same set of variables as the $\csp_1$ constraint, plus some new, auxiliary variables (these auxiliary variables are not shared between constraints of $\csp_1$).  We require that for every assignment which satisfies the $\csp_1$ constraint, there is a way to label the auxiliary variables to fully satisfy the $\csp_2$ constraints.  Furthermore, there is some parameter $0<\gadg <1$ such that for every assignment which does not satisfy the $\csp_1$ constraint, the optimum labeling to the auxiliary variables will satisfy exactly $\gadg$ fraction of the $\csp_2$ constraints.  Such a gadget reduction we call a \emph{$\gadg$-gadget-reduction} from $\csp_1$ to $\csp_2$.  The following proposition is well-known:
\begin{proposition}\label{prop:gadgets}
Suppose it is $\NP$-hard to $(c, s)$-decide $\csp_1$.  If there exists a $\gadg$-gadget-reduction from $\csp_1$ to $\csp_2$, then it is $\NP$-hard to $(c+(1-c)\gadg, s + (1-s)\gadg)$-decide $\csp_2$.
\end{proposition}
We note that the notation $\gadg$-gadget-reduction is similar to a piece of notation employed by \cite{TSSW00}, but the two have different (though related) definitions.

\subsection{Fourier analysis on $\Z_3$}\label{sec:fourier}

Let $\omega = e^{2\pi i/3}$ and set $\roots = \{\omega^0, \omega^1, \omega^2\}$.  For $\alpha \in \Z_3^n$, consider the Fourier character $\chi_\alpha :\Z_3^n \rightarrow \roots$ defined as $\chi_\alpha(x) = \omega^{\alpha\cdot x}$.  Then it is easy to see that $\E[\chi_{\alpha}(\x)\overline{\chi_{\beta}(\x)}] = {\bf 1}[\alpha = \beta]$, where here and throughout $\bx$ has the uniform probability distribution on $\Z_3^n$ unless otherwise specified..  As a result, the Fourier characters form an orthonormal basis for the set of functions $f:\Z_3^n \rightarrow \roots$ under the inner product $\la f,g \ra = \E[f(\x) g(\x)]$; i.e.,
\begin{equation*}
f = \sum_{\alpha \in \Z_3^n} \hat{f}(\alpha)\chi_\alpha,
\end{equation*}
where the $\hat{f}(\alpha)$'s are complex numbers defined as $\hat{f}(\alpha) = \E[f(\x)\overline{\chi_\alpha(\x)}]$.  For $\alpha \in \Z_3^n$, we use the notation $\vert \alpha \vert$ to denote $\sum \alpha_i$ and $\#\alpha$ to denote the number of nonzero coordinates in $\alpha$.  When $d$ is clear from context and $\alpha \in \Z_3^{dK}$, define $\pih(\alpha) \in \Z_3^K$ so that $(\pih(\alpha))_i \equiv \vert \alpha[i] \vert \pmod{3}$ (recall the notation $\alpha[i]$ from the beginning of this section).
We have Parseval's identity: for
every $f:\Z_3^n \rightarrow \roots$ it holds that $\sum_{\alpha\in\Z_3^n}\vert\hat{f}(\alpha)\vert^2 = 1$.  Note that this implies that $\vert \hat{f}(\alpha) \vert \leq 1$ for all $\alpha$, as otherwise $\vert\hat{f}(\alpha)\vert^2$ would be greater than 1.

A function $f:\Z_3^n\rightarrow \Z_3$ is said to be \emph{folded} if for every $x \in \Z_3^n$ and $c \in \Z_3$, it holds that $f(x + c) = f(x) +c$, where $(x+c)_i = x_i + c$.
\begin{proposition}\label{prop-folded}
Let $f:\Z_3^n \rightarrow \roots$ be folded.  Then $\hat{f}(\alpha)\neq 0 \Rightarrow \vert \alpha \vert \equiv 1 \pmod{3}$.
\end{proposition}
\begin{proof}
\begin{equation*}
\hat{f}(\alpha) = \E[f(\x + 1)\overline{\chi_\alpha(\x+1)}]
=\E[\omega f(\x)\overline{\chi_\alpha(\x)}\overline{\chi_\alpha(1, 1, \ldots, 1)}]
= \omega \overline{\chi_\alpha(1, 1, \ldots, 1)}\hat{f}(\alpha).
\end{equation*}
This means that $\omega\overline{\chi_\alpha(1, 1, \ldots, 1)}$ must be 1.  Expanding this quantity,
\begin{equation*}
\omega\overline{\chi_\alpha(1, 1, \ldots, 1)} = \omega^{1 - \alpha \cdot (1, 1, \ldots, 1)} = \omega^{1 - \vert \alpha \vert}.
\end{equation*}
So, $\vert \alpha \vert \equiv 1 \pmod{3}$, as promised.
\end{proof}

\subsection{Dictatorship tests}

In this paper, we make use dictatorship tests, which are a standard tool for proving $\NP$-hardness of approximation results.  Generally speaking, the input of a dictatorship test is two functions $f:\Z_3^K \rightarrow \Z_3$ and $g:\Z_3^{dK} \rightarrow \Z_3$ and a $d$-to-$1$ map $\pi:[dK] \rightarrow [K]$.  The map $\pi$ naturally groups strings $y \in \Z_3^{dK}$ into $K$ separate ``blocks'' of coordinates, the first block being the coordinates in $\pi^{-1}(1)$, the second block being the coordinates in $\pi^{-1}(2)$, etc.  Without loss of generality we will assume that $\pi$ is the map where $\pi(k) = 1$ for $1 \leq k \leq d$, $\pi(k) = 2$ for $d+1 \leq k \leq 2d$, and so on.  In this case, we write $y[i] \in \Z_3^d$ for the $i$th block of $y$, and $(y[i])_j \in \Z_3$ for the $j$th coordinate of this block.

The goal of a dictatorship test is to distinguish the case when $f$ and $g$ are ``matching dictators'' from the case when $f$ and $g$ are ``far from matching dictators''.  A function $f$ is a \emph{dictator} if $f(x) = x_i$, for some $i$.  Furthermore, $f$ and $g$ are \emph{matching dictators} if $f(x) = x_i$, $g(y) = y_j$, and $\pi(j) = i$.  In other words, they are dictators whose dictator coordinates match up according to the map $\pi$.  A  property of matching dictators is that both $f$ and $g$ ``depend on'' certain coordinates, meaning that these coordinates are important to the output of $f$ and $g$, and these coordinates match each other.  Thus, $f$ and $g$ are \emph{far from matching dictators} if there are no coordinates $i$ and $j$ which $f$ and $g$ depend on, respectively, for which $\pi(j) = i$.  An example of this is ``nonmatching'' dictators, when, say, $f(x) = x_1$ and $g(y) = y_{d+1}$.

To prove hardness for $\coloring$ (i.e., the $\neq$ constraint), one should construct a dictatorship test with the following outline: first, the test picks $\x\in \Z_3^K$ and $\y \in \Z_3^{dK}$ from some random distribution, and checks whether $f(\x) \neq g(\y)$.  If indeed this is the case, then the test passes, and otherwise it fails.  Generally, if one is interested in showing that $(c, s)$-deciding a given problem is $\NP$-hard, it suffices to construct a test for which matching dictators pass with probability at least $c$ and functions far from matching dictators pass with probability less than $s$.

We use a variant of this outline proposed in~\cite{OW12}, in which the test involves a third auxiliary function $h:S \rightarrow \Z_3$, where $S$ is some set.  We still want to distinguish the cases of $f$ and $g$ being matching dictators and functions far from matching dictators, but now the outline is a little different: in addition to selecting $\x$ and $\y$, we also select from some distribution a string $\z \in S$.  Then with some probability we test $h(\z) \neq f(\x)$ and with some probability $h(\z) \neq g(\y)$.  A test with this outline we refer to as a ``function-in-the-middle'' test, as $h$ acts as an intermediary between the functions $f$ and $g$.

\section{$\coloring$ and $\alin{2}$ tests}\label{sec:better}

In this section, we give our hardness results for $\coloring$ and 2-to-1 Label Cover, following the proof outlines described at the end of Section~\ref{sec:techniques}.  First, we state a pair of simple gadget reductions:

\begin{lemma}\label{lem:fournat-alin}
There is a $3/4$-gadget-reduction from $\fournat$ to $\alin{2}$.
\end{lemma}

\begin{lemma}\label{lem:alin-twotoone}
There is a $1/2$-gadget-reduction from $\alin{2}$ to $2$-to-$1$ Label Cover.
\end{lemma}
\noindent
Together with Proposition~\ref{prop:gadgets}, these imply the following corollary:
\begin{corollary}\label{cor:fournat-twotoone}
There is a $7/8$-gadget-reduction from $\fournat$ to $2$-to-$1$ Label Cover.  Thus, if it is $\NP$-hard to $(c, s)$-decide the $\fournat$ problem, then it is $\NP$-hard to $((7+c)/8, (7+s)/8)$-decide the 2-to-1 Label Cover problem.
\end{corollary}
\noindent
The gadget reduction from $\fournat$ to $\alin{2}$ relies on the simple fact that if $a, b, c, d \in \Z_3$ satisfy the $\fournat$ predicate, then there is some element of $\Z_3$ that none of them equal. 

\begin{proof}[Proof of Lemma~\ref{lem:fournat-alin}]
A $\fournat$ constraint $C$ on the variables $S = (v_1, v_2, v_3, v_4)$ is of the form
\begin{equation*}
\fournat(v_1 + k_1, v_2 + k_2, v_3 + k_3, v_4 + k_4),
\end{equation*}
where the $k_i$'s are all constants in $\Z_3$.  To create the $\alin{2}$ instance, introduce the auxiliary variable $y_C$ and add the four $\alin{2}$ equations
\begin{equation}
\label{eqn:fournat-alin}
v_i + k_i \neq y_C \pmod{3}, \quad i \in [4].
\end{equation}

If $f:S \rightarrow \Z_3$ is an assignment which satisfies the $\fournat$ constraint, then there is some $a \in \Z_3$ such that $f(v_i) + k_i \neq a \pmod{3}$ for all $i \in [4]$.  Assigning $a$ to $y_C$ satisfies all four equations \eqref{eqn:fournat-alin}.  On the other hand, if $f$ doesn't satisfy the $\fournat$ constraint, then $\{f(v_i) + k_i\}_{i \in [4]} = \Z_3$, so no assignment to $y_C$ satisfies all four equations.  However, it is easy to see that there is an assignment which satisfies three of the equations.  This gives a $\frac{3}{4}$-gadget-reduction from $\fournat$ to $\alin{2}$, which proves the lemma.
\end{proof}

The reduction from $\alin{2}$ to 2-to-1 Label Cover is the well-known constraint-variable reduction, and uses the fact that in the equation $v_i -  v_j \neq a \pmod{3}$, for any assignment to $v_j$ there are two valid assignments to $v_i$, and vice versa.

\begin{proof}[Proof of Lemma~\ref{lem:alin-twotoone}]
A $\alin{2}$ constraint $C$ on the variables $S = (v_1, v_2)$ is of the form
\begin{equation*}
v_1 - v_2 \neq a \pmod{3},
\end{equation*}
for some $a \in \Z_3$.  To create the 2-to-1 Label Cover instance, introduce the variable $y_C$ which will be labeled by one of the six possible functions $g:S\rightarrow \Z_3$ which satisfies $C$.  Finally, introduce the 2-to-1 constraints $y_C(v_1) = v_1$ and $y_C(v_2) = v_2$.  Here $v_1$ and $v_2$ are treated on the left as inputs to the function labeling $y_C$ and on the right as variables to be labeled with values in $\Z_3$.

If $f : S \rightarrow \Z_3$ is an assignment which satisfies the $\alin{2}$ constraint, then we label $y_C$ with $f$.  In this case,
\begin{equation*}
y_C(v_i) = f(v_i),\quad i = 1, 2.
\end{equation*}
Thus, both equations are satisfied.  On the other hand, if $f$ does not satisfy the $\alin{2}$ constraint, then any $g$ which $y_C$ is labeled with disagrees with $f$ on at least one of $v_1$ or $v_2$.  It is easy to see, though, that a $g$ can be selected to satisfy one of the two equations.  This gives a $\frac{1}{2}$-gadget-reduction from $\alin{2}$ to 2-to-1, which proves the lemma.
\end{proof}

\subsection{Three tests}

Now that we have shown that $\alin{2}$ hardness results translate into 2-to-1 Label Cover hardness results, we present our $\alin{2}$ function test.  From here, the $\coloring$ function test follows immediately.  Finally, we will show how in the course of analyzing the $\alin{2}$ test one is lead naturally to our $\fournat$ test.  This correspondence between the $\alin{2}$ test and the $\fournat$ test parallels the gadget reduction from Lemma~\ref{lem:fournat-alin}.  The test is:

\begin{center}
\framebox{$\alin{2}$ Test}
\end{center}
\qquad Given folded functions $f : \Z_3^{K} \rightarrow \Z_3$, $g, h:\Z_3^{dK} \rightarrow \Z_3$:
\begin{itemize}

\item Let $\x \in \Z_3^K$ and $\y \in \Z_3^{dK}$ be independent and uniformly random.
\item For each $i \in [K], j \in [d]$, select $(\z[i])_j$ independently and uniformly from the elements of $\Z_3\setminus\{\x_i, (\y[i])_j\}$.
\item With probability $\frac14$, test $f(\x) \neq h(\z)$; with probability $\frac34$, test $g(\y) \neq h(\z)$.
\end{itemize}

\myfig{.75}{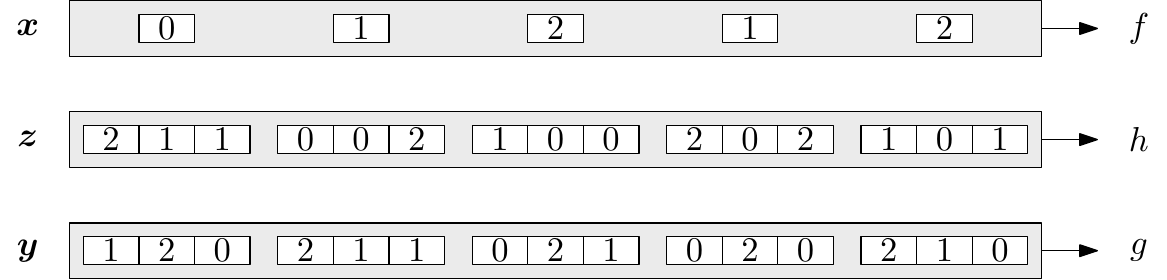}{An illustration of the $\alin{2}$ test distribution; $d = 3$, $K = 5$}{fig:test}

Above is an illustration of the test.  We remark that for any given block $i$, $z[i]$ determines $x_i$ (with very high probability), because as soon as $z[i]$ contains two distinct elements of $\Z_3$, $x_i$ must be the third element of $\Z_3$.  Notice also that in every column of indices, the input to $h$ always differs from the inputs to both $f$ and $g$.  Thus, ``matching dictator'' assignments pass the test with probability~$1$.  (This is the case in which $f(x) = x_i$ and $g(y) = (y[i])_j$ for some $i \in [K]$, $j \in [d]$.) On the other hand, if $f$ and $g$ are ``nonmatching dictators'', then they succeed with only $\frac{11}{12}$ probability.  This turns out to be essentially optimal among functions $f$ and $g$ without ``matching influential coordinates/blocks''.  We will obtain the following theorem:
\begin{named}{Theorem \ref{thm:alin-hardness} restated}
For all $\eps > 0$, it is $\NP$-hard to $(1, \frac{11}{12} + \eps)$-decide the $\alin{2}$ problem.
\end{named}

We would like to use a similar test for our $\coloring$ hardness result, but we can no longer assume that the functions $f$, $g$, and $h$ are folded.  This is problematic, as without this guarantee $f$ and $g$ could both be identically $0$ and $h$ could be identically $1$, in which case the three functions would pass the test with probability~$1$.  Since constant functions cannot be decoded to Label Cover solutions, we would like to prevent this from happening.  Thus, we will add ``folding tests'' to force $f$ and $g$ to look folded.  Having ensured this, we are free to run the $\alin{2}$ test without worry.  The test is:

\begin{center}
\framebox{$\coloring$ Test}
\end{center}
\qquad Given functions $f : \Z_3^{K} \rightarrow \Z_3$, $g, h:\Z_3^{dK} \rightarrow \Z_3$:
\begin{itemize}

\item Let $\x \in \Z_3^K$ and $\y \in \Z_3^{dK}$ be independent and uniformly random.
\item With probability $\frac{1}{17}$, test $f(\x) \neq f(\x+1)$; with probability $\frac{4}{17}$, test $g(\y) \neq g(\y+1)$.
\item With the remaining $\frac{12}{17}$ probability, run the ``non-folded'' version of the $\alin{2}$ test on $f$, $g$, and $h$.
\end{itemize}

Here, by the ``non-folded'' version of the $\alin{2}$ test, we mean the test which is identical to the $\alin{2}$ test, only it does not assume $f$, $g$, and $h$ are folded.  If $f$ and $g$ are matching dictators, then they always pass the folding tests, so as before they succeed with probability~$1$.  If on the other hand $f$ and $g$ are nonmatching dictators, then they also always pass the folding tests, so they succeed with probability $\frac{5}{17} + \frac{12}{17} \cdot \frac{11}{12} = \frac{16}{17}$.  Just as before, this turns out to be basically optimal among functions without matching influential coordinates:

\begin{named}{Theorem \ref{thm:3-coloring-hardness} restated}
For all $\eps > 0$, it is $\NP$-hard to $(1, \frac{16}{17} + \eps)$-decide the $\coloring$ problem.
\end{named}

Let us further discuss the $\alin{2}$ test.  Given $\x$, $\y$, and $\z$ from the $\alin{2}$ test, consider the following method of generating two additional strings $\y', \y'' \in \Z_3^{dK}$ which represent $h$'s ``uncertainty'' about $\y$.  For $j \in [d]$, if $\x_i = (\y[i])_j$, then set both $(\y'[i])_j$ and $(\y''[i])_j$ to the lone element of $\Z_3 \setminus \{\x_i, (\z[i])_j\}$.  Otherwise, set one of $(\y'[i])_j$ or $(\y''[i])_j$ to $\x_i$, and the other one to $(\y[i])_j$.  It can be checked that $\buddy(\x_i, (\y[i])_j, (\y'[i])_j, (\y''[i])_j) =  1$, a more stringent requirement than satisfying $\fournat$.  In fact, the marginal distribution on these four variables is a uniformly random assignment that satisfies the $\buddy$ predicate.  Conditioned on $\x$ and $\z$, the distribution on $\y'$ and $\y''$ is identical to the distribution on $\y$.  To see this, first note that by construction, neither $(\y'[i])_j$ nor $(\y''[i])_j$ ever equals $(\z[i])_j$.  Further, because these indices are distributed as uniformly random satisfying assignments to $\buddy$, $\Pr[(\y'[i])_j = x_i] = \Pr[(\y''[i])_j = x_i] = \frac13$, which matches the corresponding probability for $\y$. Thus, as $\y$, $\y'$, and $\y''$ are distributed identically, we may rewrite the test's success probability as:
\begin{align}
\Pr[\text{$f$, $g$, and $h$ pass the $\alin{2}$ test}]
& = \tfrac14\Pr[f(\x) \neq h(\z)] + \tfrac34\Pr[g(\y) \neq h(\z)]\nonumber\\
& = \text{avg}\left\{
			\begin{array}{r}
				\Pr[f(\x) \neq h(\z) ], \\
				\Pr[g(\y) \neq h(\z) ], \\
				\Pr[g(\y') \neq h(\z) ], \\
				\Pr[g(\y'') \neq h(\z) ]\phantom{,}
			\end{array} \right\}\nonumber\\
&\leq \frac34 + \frac14 \E[\fournat(f(\x), g(\y), g(\y'), g(\y''))]. \label{eq:hidden-gadget}
\end{align}
This is because if $\fournat$ fails to hold on the tuple $(f(\x), g(\y), g(\y'), g(\y''))$, then $h(\z)$ can disagree with at most~$3$ of them.

At this point, we have removed $h$ from the test analysis and have uncovered what appears to be a hidden $\fournat$ test inside the $\alin{2}$ test: simply generate four strings $\x$, $\y$, $\y'$, and $\y''$ as described earlier, and test $\fournat(f(\x), g(\y), g(\y'), g(\y''))$.  With some renaming of variables, this is exactly what our $\fournat$ test does:

\begin{center}
\framebox{$\fournat$ Test}
\end{center}
\qquad Given folded functions $f : \Z_3^{K} \rightarrow \Z_3$, $g:\Z_3^{dK} \rightarrow \Z_3$:
\begin{itemize}
\item Let $\x \in \Z_3^K$ be uniformly random.
\item Select $\y, \z, \w$ as follows: for each $i \in [K], j \in [d]$, select $((\y[i])_j, (\z[i])_j, (\w[i])_j)$ uniformly at random from the elements of $\Z_3$ satisfying $\buddy(\x_i, (\y[i])_j,  (\z[i])_j, (\w[i])_j)$.
\item Test $\fournat(f(\x), g(\y), g(\z), g(\w))$.
\end{itemize}

\myfig{.75}{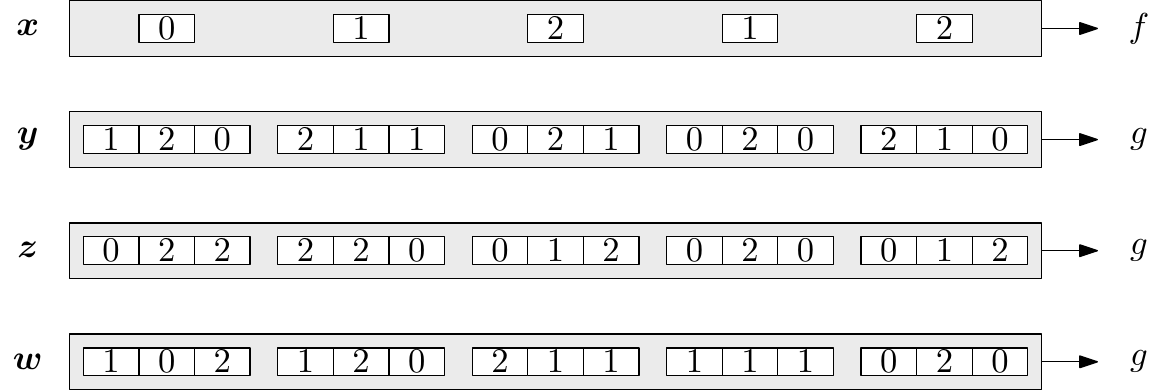}{An illustration of the $\fournat$ test distribution; $d = 3$, $K = 5$}{fig:nat-test}

Above is an illustration of this test.  In this illustration, the strings $\z$ and $\w$ were derived from the strings in Figure~\ref{fig:test} using the process detailed above for generating $\y'$ and $\y''$.  Note that each column is missing one of the elements of $\Z_3$, and that each column satisfies the $\buddy$ predicate.  Because satisfying $\buddy$ implies satisfying $\fournat$, matching dictators pass this test with probability~$1$.  On the other hand, it can be seen that nonmatching dictators pass the test with probability $\frac23$.  This is basically optimal among functions with no matching influential coordinates:
\begin{named}{Theorem \ref{thm:fournat-hardness} restated}
For all $\eps > 0$, it is $\NP$-hard to $(1, \frac{2}{3} + \eps)$-decide the $\fournat$ problem.
\end{named}

Unfortunately, it is not clear if there is a similar gadget reducing $\coloring$ to $\fournat$, or to any other simple 4CSP for that matter.  However, by using Equation~\eqref{eq:hidden-gadget}, we can still reduce the analysis of the $\coloring$ test to analyzing the $\fournat$ test:
\begin{align}
\Pr[\text{$f$, $g$, and $h$ pass the $\coloring$ test}]
\leq& \frac{1}{17} \cdot p_f + \frac{4}{17} \cdot p_g\nonumber\\
&+ \frac{12}{17}\cdot\left(\frac34 + \frac14 \E[\fournat(f(\x), g(\y), g(\z), g(\w))]\right).\label{eq:3-coloring-to-4nat}
\end{align}
Here $p_f$ and $p_g$ are the probabilities that $f$ and $g$ pass the folding test, respectively, and $\x$, $\y$, $\z$, and $\w$ are distributed as in the $\fournat$ test.  This equation will be the focus of our $\coloring$ soundness proof.

(As one additional remark, our $\alin{2}$ test is basically the composition of the $\fournat$ test with the gadget from Lemma~\ref{lem:fournat-alin}.  In this test, if we instead performed the $f(\x) \neq h(\z)$ test with probability $\frac13$ and the $g(\y) \neq h(\z)$ test with probability $\frac23$, then the resulting test would basically be the composition of a $\alin{3}$ test with a suitable $\alin{3}$-to-$\alin{2}$ gadget.)

\section{Fourier analysis}\label{sec:analysis}

Let $\omega = e^{2\pi i/3}$, and set $\roots = \{\omega^0, \omega^1, \omega^2\}$.  In what follows, we identify $f$ and $g$ with the functions $\omega^f$ and $\omega^g$, respectively, whose range is $\roots$ rather than $\Z_3$.  Set $L = dK$.  Define
\begin{equation*}
\decode(f, g) := \sum_{\alpha: \pih(\alpha) \neq 0}\vert\hat{f}(\pih(\alpha))\vert\cdot \vert\hat{g}(\alpha)\vert^2\cdot(1/2)^{\#\alpha}.
\end{equation*}
This quantity corresponds to the ``decodable'' part of $f$ and $g$.  
This section is devoted to proving the following two lemmas:
\begin{lemma}\label{lem:4nat-fourier}
Let $f:\Z_3^K\rightarrow \roots$ and $g:\Z_3^{L}\rightarrow \roots$ be folded.  Then the probability $f$ and $g$ pass the $\fournat$ test is at most $\frac{2}{3} + \frac{2}{3} \decode(f, g)$.
\end{lemma}
\begin{lemma}\label{lem:3-coloring-fourier}
Let $f:\Z_3^K\rightarrow \roots$ and $g:\Z_3^L \rightarrow \roots$.  Then the probability $f$ and $g$ pass the $\coloring$ test is at most $\frac{16}{17} + \frac{2}{17}\decode(f, g)$.
\end{lemma}
\noindent Having proven these, our hardness results follow immediately from a standard application of \Hastad's method.  See Section~\ref{app:fournat} for details.

The first step is to ``arithmetize'' the $\fournat$ predicate.  It is not hard to verify that
\begin{align*}
  \fournat(a_1, a_2, a_3, a_4) &= \frac{5}{9} + \frac{1}{9} \sum_{i \ne j} \omega^{a_i} \overline{\omega}^{a_j} - \frac{1}{9} \sum_{i < j < k} \omega^{a_i} \omega^{a_j} \omega^{a_k} - \frac{1}{9} \sum_{i < j < k} \overline{\omega}^{a_i}\overline{\omega}^{a_j}\overline{\omega}^{a_k} \\
&= \frac59 + \frac{2}{9}\sum_{i<j}\Re[\omega^{a_i}\overline{\omega}^{a_j}] - \frac{2}{9}\sum_{i<j<k} \Re[\omega^{a_i}\omega^{a_j}\omega^{a_k}].
\end{align*}
Using the symmetry between $\y$, $\z$, and $\w$, we deduce
\begin{multline}
\E[\fournat(f(\x), g(\y), g(\z), g(\w))] \\= \tfrac59 + \tfrac23 \Re \E[f(\bx)\overline{g(\y)}] + \tfrac23 \Re \E[g(\by)\overline{g(\bz)}] - \tfrac23 \Re \E[f(\bx)g(\by)g(\bz)] - \tfrac29 \Re \E[g(\by)g(\bz)g(\bw)]. \label{eq:bigexpansion}
\end{multline}
To analyze this expression, we will need the following lemma:
\begin{lemma}\label{lem-downanddirty}
Let $a \in \Z_3$, $\beta, \gamma \in \Z_3^{dK}$, and $i$ and $j$ be such that $\pi(j) = i$.  Then
\begin{equation*}
\E_{\y, \z}\left[ \omega^{\beta_j \y_j + \gamma_j \z_j} \mid \x_i = a\right] = 
\left\{
	\begin{array}{cl}
		 \left(-\frac{1}{2}\right)^{\#\beta_j}\omega^{2a\beta_j} & \text{if } \beta_j = \gamma_j,\\
		0 & \text{otherwise}.
	\end{array} \right.
\end{equation*}
\end{lemma} 
\begin{proof}
Conditioned on $\x_i = a$, the distribution on the values for $(\y_j, \z_j)$ is uniform on the six possibilities $(a, a+1)$, $(a, a+2)$, $(a+1, a)$, $(a+1, a+1)$, $(a+2, a)$, and $(a+2, a+2)$. If $\beta_j = \gamma_j$, then the expectation equals $\E[\omega^{\beta_j(\y_j+\z_j)}\mid \x_i = a]$.  As either $\y_j + \z_j \equiv 2a + 1 \pmod{3}$ or $\y_j + \z_j \equiv 2a+2 \pmod{3}$, each with probability half,  this is equal to
\begin{equation*}
\frac{1}{2}\left(\omega^{\beta_j(2a+1)} + \omega^{\beta_j(2a+2)}\right)
= \frac{(\omega^{\beta_j} + \omega^{2\beta_j})}{2} \omega^{2a\beta_j}
 = \left(-\frac{1}{2}\right)^{\#\beta_j}\omega^{2a\beta_j}.
\end{equation*}

On the other hand, If $\beta_j \neq \gamma_j$, then either only one of $\beta_j$ or $\gamma_j$ is zero, or neither is zero, and $\gamma_j \equiv -\beta_j \pmod{3}$.  In the first case, the expectation is either $\E[\omega^{\beta_j\y_j}\mid \x_i = a]$ or $\E[\omega^{\gamma_j\z_j}\mid \x_i = a]$ for a nonzero $\beta_j$ or a nonzero $\gamma_j$, respectively.  Both of these expectations are zero, as both $\y_j$ and $\z_j$ are uniform on $\Z_3$.  In the second case,
\begin{align*}
\E[\omega^{\beta_j \y_j + \gamma_j \z_j} \mid \x_i = a] = &\E[\omega^{\beta_j\y_j - \beta_j\z_j} \mid \x_i = a]\\
= &\E[\omega^{\beta_j(\y_j - \z_j)} \mid \x_i = a],
\end{align*}
which is zero, because $\beta_j$ is nonzero and $\y_j - \z_j$ is uniformly distributed on $\Z_3$.
\end{proof}

Now we use this to find an expression for a general form of the $\E[f(\x)g(\y)g(\z)]$ term:
\begin{lemma}\label{lem-eggg-general}
Let $f_1:\Z_3^K \rightarrow \mathbb{R}$ and $g_1, g_2: \Z_3^L \rightarrow \mathbb{R}$.  Then
\begin{equation*}
\E[f_1(\x)g_1(\y)g_2(\z)] = \sum_{\alpha \in \Z_3^L}\hat{f}_1(\pih(\alpha))\hat{g}_1(\alpha)\hat{g}_2(\alpha)\left(-\frac{1}{2}\right)^{\#\alpha}.
\end{equation*}
\end{lemma}

From this, we can derive the following two corollaries:
\begin{corollary}\label{cor-egg}
Let $g:\Z_3^{dK} \rightarrow \mathbb{R}$.  Then
\begin{equation*}
\E[g(\y) \overline{g(\z)}]
= \sum_{\alpha : \vert \alpha[i] \vert \equiv 0~\forall i} \hat{g}(\alpha)\overline{\hat{g}(-\alpha)} \left(-\frac{1}{2}\right)^{\# \alpha}.
\end{equation*}
\end{corollary}
\begin{proof}[Proof (assuming Lemma~\ref{lem-eggg-general}).]  Set $f_1 \equiv 1$, $g_1 = g$, and $g_2 = \overline{g}$.  The only nonzero Fourier coefficient of $f_1$ is $\hat{f}_1(0) = 1$, and the only elements $\alpha \in \Z_3^L$ for which $\pi_3(\alpha) = 0$ are those where $\vert \alpha[i] \vert \equiv 0$ for all $i$.
Apply Lemma~\ref{lem-eggg-general} to these three functions:
\begin{equation*}
\E[g(\y) \overline{g(\z)}] = \sum_{\alpha : \vert \alpha[i] \vert \equiv 0~\forall i} \hat{g}_1(\alpha){\hat{g}_2(\alpha)} \left(-\frac{1}{2}\right)^{\# \alpha}.
\end{equation*}
Since $\hat{g}_1(\alpha) = \hat{g}(\alpha)$, it remains to show that $\hat{g}_2(\alpha) = \overline{\hat{g}(-\alpha)}$, and this is true because
\begin{equation*}
\hat{g}_2(\alpha)
= \E[g_2(\y)\overline{\chi_\alpha(\y)}]
= \E[\overline{g(\y)}\overline{\chi_\alpha(\y)}]
= \E[\overline{g(\y)}\chi_{-\alpha}(\y)]
= \overline{\E[g(\y)\overline{\chi_{-\alpha}(\y)}]}
= \overline{\hat{g}(-\alpha)},
\end{equation*}
where the third equality follows from $\overline{\chi_\beta(\z)} = \overline{\omega^{\beta \cdot \z}} = \omega^{-\beta \cdot \z} = \chi_{-\beta}(\z)$. 
\end{proof}

\begin{corollary}\label{cor-efgg}
$-\Re\E[f(\x)g(\y)g(\z)] \leq  \decode(f, g) + \vert \hat{f}(0) \vert \sum_{\alpha: \pih(\alpha) = 0} \vert\hat{g}(\alpha)\vert^2\cdot(1/2)^{\#\alpha}$.
\end{corollary}
\begin{proof}[Proof (assuming Lemma~\ref{lem-eggg-general}).]  Applying Lemma~\ref{lem-eggg-general} to $f$, $g$, and $g$:
\begin{align*}
-\Re\E[f(\x)g(\y)g(\z)] 
& = -\Re\sum_{\alpha \in \Z_3^L}\hat{f}(\pih(\alpha))\hat{g}(\alpha)^2\left(-\frac{1}{2}\right)^{\#\alpha}\\
& \leq \sum_{\alpha \in \Z_3^L}\vert\hat{f}(\pih(\alpha))\vert\cdot \vert\hat{g}(\alpha)\vert^2\cdot(1/2)^{\#\alpha}\\
& =  \decode(f, g) +  \sum_{\alpha: \pih(\alpha) = 0} \vert \hat{f}(0) \vert \cdot \vert\hat{g}(\alpha)\vert^2\cdot(1/2)^{\#\alpha}.\qedhere
\end{align*}
\end{proof}

We now prove Lemma~\ref{lem-eggg-general}.
\begin{proof}[Proof of Lemma~\ref{lem-eggg-general}]
Begin by expanding out $\E[f_1(\x)g_1(\y)g_2(\z)]$:
\begin{equation}
\E[f_1(\x)g_1(\y)g_2(\z)] = \sum_{ \alpha \in \Z_3^K, \beta, \gamma \in \Z_3^L}
\hat{f}_1(\alpha)\hat{g}_1(\beta)\hat{g}_2(\gamma)\E[\chi_\alpha(\x)\chi_{\beta}(\y)\chi_\gamma(\z)].\label{eq:fourierexpanded}
\end{equation}
We focus on the products of the Fourier characters:
\begin{equation}
\E[\chi_\alpha(\x)\chi_\beta(\y)\chi_\gamma(\z)]
=  \prod_{i \in [K]}\E[\chi_{\alpha_i}(\x_i) \chi_{\beta[i]}(\y[i])\chi_{\gamma[i]}(\z[i])]\label{eq:fcharacters}
\end{equation}
We can attend to each block separately:
\begin{align}
\E[\chi_{\alpha_i}(\x_i) \chi_{\beta[i]}(\y[i])\chi_{\gamma[i]}(\z[i])]
= &\E\left[\omega^{\alpha_i \cdot \x_i + \beta[i] \cdot \y[i] + \gamma[i] \cdot \z[i]}\right]\nonumber\\
= &\E_{\x}\left[\omega^{\alpha_i \cdot a}\prod_{j:\pi(j) = i}\underbrace{\E_{\y, \z}\left[\omega^{\beta_j \y_j + \gamma_j \z_j} \mid \x_i = a\right]}_{(*)}\right].\label{eq:blockcharacters}
\end{align}

Lemma~\ref{lem-downanddirty} tells us that the expectation $(*)$ is zero if $\beta_j \neq \gamma_j$.  Thus, if Equation~\eqref{eq:fcharacters} is to be nonzero, it must be the case that $\beta = \gamma$.  If this is the case, then we can rewrite Equation~\eqref{eq:blockcharacters} as
\begin{equation*}
\eqref{eq:blockcharacters}
=  \E_{\x}\left[\omega^{\alpha_i \cdot a} \prod_{j:\pi(j) = i} \left(-\frac{1}{2}\right)^{\#\beta_j}\omega^{2a\beta_j}\right]
=  \E_{\x}\left[\left(-\frac{1}{2}\right)^{\#\beta[i]}\omega^{(\alpha_i  + 2 \vert \beta[i]\vert)a}\right].
\end{equation*}
Note that the exponent of $\omega$, $(\alpha_i +2\vert \beta[i]\vert) a$, is zero if $\alpha_i \equiv \vert \beta[i]\vert \pmod{3}$, in which case the expectation is just the constant $(-1/2)^{\#\beta[i]}$.  This occurs for all $i \in [K]$ exactly when $\alpha = \pih(\beta)$.  If, on the other hand, $\alpha_i + 2\vert\beta[i]\vert$ is nonzero, then the entire expectation is zero because $a$, the value of $\x_i$, is uniformly random from $\Z_3$.  Thus, Equation \eqref{eq:fcharacters} is nonzero only when $\alpha = \pih(\beta)$ and $\beta = \gamma$, in which case it equals
\begin{equation*}
\eqref{eq:fcharacters} = \left(-\frac{1}{2}\right)^{\#\beta}.
\end{equation*}
We may therefore conclude with
\begin{equation*}
\eqref{eq:fourierexpanded} = \sum_{\alpha \in \Z_3^L}\hat{f}_1(\pih(\alpha))\hat{g}_1(\alpha)\hat{g}_2(\alpha)\left(-\frac{1}{2}\right)^{\#\alpha}. \qedhere
\end{equation*}
\end{proof}

\subsection{$\fournat$ Analysis}

In the $\fournat$ test, we may assume that $f$ and $g$ are folded, which immediately implies that $\E[f(\x)\overline{g(\y)}] = 0$.  This is because $\x$ and $\y$ are independent, and hence \begin{equation*}
\E[f(\x)\overline{g(\y)}] = \E[f(\x)]\E[\overline{g(\y)}] = 0 \cdot 0
\end{equation*}
since $f$ and $g$ are folded.  Next, folding also implies that $\E[g(\y)\overline{g(\z)}] = 0$.  To see this, first note that for any $\alpha$ for which $\vert \alpha[i] \vert \equiv 0$ for all $i$, we have that $\vert \alpha \vert \equiv 0$.  Thus, any such $\alpha$ must satisfy $\hat{g}(\alpha) = 0$, as Proposition~\ref{prop-folded} implies that $\hat{g}(\alpha') \neq 0$ only when $\vert \alpha' \vert \equiv 1$.  This means the sum in Corollary~\ref{cor-egg} must be zero, which implies that $\E[g(\y)\overline{g(\z)}] = 0$ as well.

Equation~\eqref{eq:bigexpansion} has now been reduced to
\begin{equation}
\eqref{eq:bigexpansion} = \tfrac{5}{9} - \tfrac{2}{3} \Re \E[f(\x)g(\y)g(\z)] - \tfrac{2}{9}\Re \E[g(\y)g(\z)g(\w)]. \label{eq:smallexpansion}
\end{equation}
As $g(\y)g(\z)g(\w)$ is always in $\roots$, $\Re \E[g(\y)g(\z)g(\w)]$ is always at least $-\frac12$.  Therefore,
\begin{equation}
\eqref{eq:smallexpansion}
\leq \tfrac{2}{3} - \tfrac{2}{3} \Re \E[f(\x)g(\y)g(\z)]
= \frac{2}{3} + \frac{2}{3} \decode(f, g),
\end{equation}
using Corollary~\ref{cor-efgg} and the fact that $\hat{f}(0) = 0$ by folding.  This proves Lemma~\ref{lem:4nat-fourier}.

\subsection{3-Coloring Analysis}

The analysis of the 3-Coloring test is more involved, partially because we can no longer assume either of the functions are folded, and partially because we need a more careful analysis of the $\E[g(\y)g(\z)g(\w)]$ term.  Instead, we upper-bound these terms with expressions involving the empty coefficients $\hat{f}(0)$ and $\hat{g}(0)$, which, when large, cause the folding tests to fail with high probability.  In addition, the analysis of the 3-Coloring test also involves analyzing the folding tests on $f$ and $g$, and it is with these that we start.

For a function $f_1:\Z_3^n\rightarrow \roots$, define $\even(f_1) = \sum_{\alpha: \vert \alpha \vert \equiv 0} \vert \hat{f}_1(\alpha) \vert^2$.
\begin{lemma}\label{lem:folding}
$\Pr[f(\x) \neq f(\x+1)] = 1-\even(f)$.
\end{lemma}
\begin{proof}
It is easy to see that
$\Pr[f(\x) \neq f(\x+1)] = \frac{2}{3}\left(1 - \Re\E[f(\x)\overline{f(\x+1)}]\right)$.
Expanding the expectation,
\begin{align*}
\E[f(\x)\overline{f(\x+1)}]
& = \sum_{\alpha, \beta} \hat{f}(\alpha) \overline{\hat{f}(\beta)} \E[\chi_\alpha(\x)\overline{\chi_\beta(\x+1)}]\\
& = \sum_{\alpha, \beta} \hat{f}(\alpha)\overline{\hat{f}(\beta)} \E\left[\chi_\alpha(\x)\overline{\chi_\beta(\x)} \omega^{-\vert \beta\vert}\right].
\end{align*}
Since $\E[\chi_{\alpha}(\x)\overline{\chi_{\beta}(\x)}] = \bone[\alpha = \beta]$, this equals $\sum_{\alpha} \vert \hat{f}(\alpha) \vert^2 \omega^{-\vert\alpha\vert}$.  Taking the real part,
\begin{align*}
\Re \sum_{\alpha} \vert \hat{f} \vert^2 \omega^{-\vert \alpha \vert}
& = \even(f) - \frac{1}{2} \sum_{\vert \alpha \vert \not\equiv 0} \vert \hat{f} \vert^2\\
& = \even(f) - \frac{1}{2}(1-\even(f)) = \frac{3}{2} \even(f) -\frac{1}{2}.
\end{align*}
Thus, the probability of passing the folding test is $\frac{2}{3}\left(1 - \frac{3}{2} \even(f) + \frac{1}{2}\right) = 1 - \even(f)$.
\end{proof}

Now we focus on the $\E[\fournat(\cdots)]$ term.  Let us upper-bound the terms in Equation~\eqref{eq:bigexpansion} from left to right.  First,
\begin{proposition}\label{prop:efg-coloring}
$\Re\E[f(\x)\overline{g(\y)}] \leq \frac{1}{2}(\vert \hat{f}(0) \vert^2 + \vert \hat{g}(0) \vert^2)$.
\end{proposition}
\begin{proof}
By the independence of $\x$ and $\y$, $\E[f(\x)\overline{g(\y)}] = \E[f(\x)] \cdot \E[\overline{g(\y)}] = \hat{f}(0) \overline{\hat{g}(0)}$.  Then,
\begin{equation*}
\Re \hat{f}(0) \overline{\hat{g}(0)}
\leq \vert \hat{f}(0) \vert \cdot \vert \hat{g}(0) \vert
\leq \frac{1}{2}(\vert \hat{f}(0) \vert^2 + \vert \hat{g}(0) \vert^2),
\end{equation*}
using the fact that $a^2 + b^2 \geq 2 a b$ for all real numbers $a$ and $b$.
\end{proof}
Next,
\begin{lemma}\label{lem:egg-coloring}
$\Re \E[g(\y)\overline{g(\z)}]
\leq \even(g)$.
\end{lemma}
\begin{proof}
From Corollary~\ref{cor-egg}, 
\begin{equation*}
\Re \E[g(\y) \overline{g(\z)}]
= \Re \sum_{\alpha : \vert \alpha[i] \vert \equiv 0~\forall i} \hat{g}(\alpha)\overline{\hat{g}(-\alpha)} \left(-\frac{1}{2}\right)^{\# \alpha}
\leq \sum_{\alpha : \vert \alpha[i] \vert \equiv 0~\forall i} \vert \hat{g}(\alpha)\vert \vert \hat{g}(-\alpha) \vert \left(\frac{1}{2}\right)^{\# \alpha}.
\end{equation*}
By Cauchy-Schwarz, this is at most
\begin{equation*}
\sqrt{ \sum_{\alpha : \vert \alpha[i] \vert \equiv 0~\forall i} \vert \hat{g}(\alpha)\vert^2  \left(\frac{1}{2}\right)^{\# \alpha}}
\cdot \sqrt{ \sum_{\alpha : \vert \alpha[i] \vert \equiv 0~\forall i} \vert \hat{g}(-\alpha) \vert^2 \left(\frac{1}{2}\right)^{\# \alpha}}
= \sum_{\alpha : \vert \alpha[i] \vert \equiv 0~\forall i} \vert \hat{g}(\alpha)\vert^2  \left(\frac{1}{2}\right)^{\# \alpha},
\end{equation*}
which is clearly at most $\sum_{\vert \alpha \vert \equiv 0} \vert \hat{g}(\alpha) \vert^2 = \even(g)$.
\end{proof}
Next,
\begin{lemma}\label{lem:efg-coloring}
$-\Re\E[f(\x)g(\y)g(\z)] 
\leq \vert \hat{f}(0) \vert \cdot \left(\frac{3}{4} \vert \hat{g}(0) \vert^2 + \frac{1}{4} \even(g)\right) +  \decode(f, g).$
\end{lemma}
\begin{proof}
By Corollary~\ref{cor-efgg},
\begin{equation*}
-\Re\E[f(\x)g(\y)g(\z)] \leq \decode(f, g) + \vert \hat{f}(0) \vert \sum_{\alpha: \pih(\alpha) \equiv 0} \vert\hat{g}(\alpha)\vert^2\cdot(1/2)^{\#\alpha}.
\end{equation*}
Consider the sum $\sum_{\alpha: \pih(\alpha) \equiv 0} \vert\hat{g}(\alpha)\vert^2\cdot(1/2)^{\#\alpha}$.  The only time that $\# \alpha = 0$ is when $\alpha = 0$.  In addition, no $\alpha$ with $\# \alpha = 1$ contributes to the sum, because such an $\alpha$ cannot satisfy $\pi_3(\alpha) \equiv 0$ (one of its coordinates must be~$1$ or~$2$).  Thus, the sum is upper-bounded by
\begin{align*}
\vert \hat{g}(0) \vert^2 + \frac{1}{4} \sum_{\alpha: \pi_3(\alpha) \equiv 0} \vert \hat{g}(\alpha) \vert^2
&\leq \vert \hat{g}(0) \vert^2 + \frac{1}{4} \sum_{\vert \alpha\vert = 0} \vert \hat{g}(\alpha) \vert^2\\
&= \vert \hat{g}(0) \vert^2 + \frac{1}{4} (\even(g) - \vert \hat{g}(0)\vert^2)
= \frac{3}{4} \vert \hat{g}(0) \vert^2 + \frac{1}{4}\even(g).
\end{align*}
This concludes the lemma.
\end{proof}

The last term, $\E[g(\y)g(\z)g(\w)]$, is more difficult to bound.  The bound we use is:
\begin{lemma}\label{lem:eggg}
$-\Re\E[g(\y)g(\z)g(\w)] \leq \frac{1}{2} - \frac{3}{2}\vert\hat{g}(0)\vert^2$.
\end{lemma}
\noindent The proof of Lemma~\ref{lem:eggg} is presented in Appendix~\ref{sec:eggg}.

Substituting Proposition~\ref{prop:efg-coloring} and Lemmas~\ref{lem:egg-coloring}, \ref{lem:efg-coloring}, and~\ref{lem:eggg} into Equation~\eqref{eq:bigexpansion} and performing some arithmetic yields
\begin{equation*}
\E[\fournat(\cdots)] 
\leq \frac{2}{3} + \frac{2}{3} \decode(f, g) + \frac{1}{3} \vert \hat{f}(0)\vert^2 + \frac{1}{2} \vert \hat{f}(0)\vert \vert\hat{g}(0) \vert^2 + \even(g)\left(\frac{2}{3} + \frac{\vert \hat{f}(0)\vert}{6}\right).
\end{equation*}
By plugging this bound into Equation~\ref{eq:3-coloring-to-4nat}, applying Lemma~\ref{lem:folding}, and performing more arithmetic, we can upper bound the probability that $f$ and $g$ pass the 3-Coloring test by
\begin{equation*}
-\frac{1}{17} \even(f) - \even(g)\left(\frac{2}{17} - \frac{\vert \hat{f}(0) \vert}{34}\right)
+ \frac{1}{17} \vert \hat{f} (0) \vert^2 + \frac{3}{34} \vert \hat{f} (0) \vert \vert \hat{g}(0) \vert^2 
+ \frac{2}{17} \decode(f, g) + \frac{16}{17}.
\end{equation*}
Note that because $0 \leq \vert \hat{f} (0) \vert \leq 1$, the coefficient of $\even(g)$ is always negative.  Thus, we may bound $-\even(f)$ and $-\even(g)$ by $-\vert \hat{f}(0) \vert^2$ and $-\vert \hat{g} (0) \vert^2$, respectively, resulting in a total upper bound of
\begin{equation*}
\frac{2}{17}\left(\vert \hat{f}(0) \vert \vert\hat{g}(0)\vert^2 - \vert \hat{g}(0)\vert^2\right) + \frac{2}{17}\decode(f, g) + \frac{16}{17}.
\end{equation*}
The leftmost term is always at most zero, so this is at most $\frac{2}{17}\decode(f, g) + \frac{16}{17}$, the expression claimed in Lemma~\ref{lem:3-coloring-fourier}.

\section{Hardness of $\fournat$}\label{app:fournat}

In this section, we show the following theorem:
\begin{named}{Theorem \ref{thm:fournat-hardness} (detailed)}
For all $\eps > 0$, it is $\NP$-hard to $(1, \frac23 + \eps)$-decide the $\fournat$ problem.  In fact, in the ``yes case'', all $\fournat$ constraints can be satisfied by $\twopair$ assignments.
\end{named}
Combining this with Lemma~\ref{lem:fournat-alin} yields Theorem~\ref{thm:alin-hardness}, and combining this with Corollary~\ref{cor:fournat-twotoone} yields Theorem~\ref{thm:twotoone-hardness}.  It is not clear whether this gives optimal hardness assuming perfect completeness.  The $\fournat$ predicate is satisfied by a uniformly random input with probability $\frac59$, and by the method of conditional expectation this gives a deterministic algorithm which $(1, \frac59)$-approximates the $\fournat$ CSP.  This leaves a gap of $\frac19$ in the soundness, and to our knowledge there are no better known algorithms.

On the hardness side, consider a uniformly random satisfying assignment to the $\buddy$ predicate.  It is easy to see that each of the four variables is assigned a uniformly random value from $\Z_3$, and also that the variables are pairwise independent.  As any satisfying assignment to the $\buddy$ predicate also satisfies the $\fournat$ predicate, the work of Austrin and Mossel~\cite{AM09} immediately implies that $(1-\eps, \frac59+\eps)$-approximating the $\fournat$ problem is $\NP$-hard under the Unique Games conjecture.  Thus, if we are willing to sacrifice a small amount in the completeness, we can improve the soundness parameter in Theorem~\ref{thm:fournat-hardness}.  Whether we can improve upon the soundness without sacrificing perfect completeness is open.

We now arrive at the proof of Theorem~\ref{thm:fournat-hardness}.  The proof is entirely standard, and proceeds by reduction from $d$-to-1 Label Cover.  A nearly identical proof gives Theorem~\ref{thm:3-coloring-hardness}, which we omit.  The proof makes use of our analysis of the $\fournat$ test, which is presented in Section~\ref{sec:analysis}.  One preparatory note: most of the proof concerns functions $f:\Z_3^K\rightarrow \Z_3$ and $g:\Z_3^{dK}\rightarrow \Z_3$.  However, we also be making use of Fourier analytic notions defined in Section~\ref{sec:fourier}, and this requires dealing with functions whose range is $\roots$ rather than $\Z_3$.  Thus, we associate $f$ and $g$ with the functions $\omega^f$ and $\omega^g$, and whenever Fourier analysis is used it will actually be with respect to the latter two functions.

\begin{proof}
Let $G = (U \cup V, E)$ be a $d$-to-1 Label Cover instance with alphabet size $K$ and $d$-to-1 maps $\pi_e:[dK]\rightarrow [K]$ for each edge $e \in E$.  We construct a $\fournat$ instance by replacing each vertex in $G$ with its Long Code and placing constraints on adjacent Long Codes corresponding to the tests made in the $\fournat$ test.  Thus, each $u\in U$ is replaced by a copy of the hypercube $\Z_3^K$ and labeled by the function $f_u:\Z_3^K \rightarrow \Z_3$.  Similarly, each $v \in V$ is replaced by a copy of the Boolean hypercube $\Z_3^{dK}$ and labeled by the function $g_v:\Z_3^{dK} \rightarrow \Z_3$.  Finally, for each edge $\{u, v\} \in E$, a set of $\fournat$ constraints is placed between $f_u$ and $g_v$ corresponding to the constraints made in the $\fournat$ test, and given a weight equal to the probability the constraint is tested in the $\fournat$ test multiplied by the weight of $\{u,v\}$ in $G$.  This produces a $\fournat$ instance whose weights sum to 1 which is equivalent to the following test:
\begin{itemize}
\item Pick an edge $e=(u, v) \in E$ uniformly at random.
\item Reorder the indices of $g_v$ so that the $k$th group of $d$ indices corresponds to $\pi_e^{-1}(k)$.
\item Run the $\fournat$ test on $f_u$ and $g_v$.  Accept iff it does.
\end{itemize}

\paragraph{Completeness}
If the original Label Cover instance is fully satisfiable, then there is a function $F:U\cup V \rightarrow [dK]$ for which $\mathsf{val}(F)=1$.  Set each $f_u$ to the dictator assignment $f_u(x) = x_{F(u)}$ and each $g_v$ to the dictator assignment $g_v(y) = y_{F(v)}$.  Let $e=\{u, v\} \in E$.  Because $F$ satisfies the constraint $\pi_e$, $F(u) = \pi_e(F(v))$.  Thus, $f_u$ and $g_v$ correspond to ``matching dictator'' assignments, and above we saw that matching dictators pass the $\fournat$ test with probability 1.  As this applies to every edge in $E$, the $\fournat$ instance is fully satisfiable.

\paragraph{Soundness}
Assume that there are functions $\{f_u\}_{u \in U}$ and $\{g_v\}_{v \in V}$ which satisfy at least a $\frac23+\eps$ fraction of the $\fournat$ constraints.  Then there is at least an $\eps/2$ fraction of the edges $e =\{u, v\} \in E$ for which $f_u$ and $g_v$ pass the $\fournat$ test with probability at least $\frac23+\eps/2$.  This is because otherwise the fraction of $\fournat$ constraint satisfied would be at most
\begin{equation*}
\left(1-\frac{\eps}{2}\right)\left(\frac{2}{3}+\frac{\eps}{2}\right) + \frac{\eps}{2}(1)
= \frac{2}{3} + \frac{2\eps}{3} - \frac{\eps^2}{4} < \frac{2}{3} + \eps.
\end{equation*}
Let $E'$ be the set of such edges, and consider $\{u, v\} \in E'$.  Set $L = dK$.
By Lemma~\ref{lem:4nat-fourier},
\begin{equation*}
\frac{2}{3} + \frac{\eps}{2} \leq \Pr[\text{$f_u$ and $g_v$ pass the $\fournat$ test}]
\leq \frac{2}{3} + \frac{2}{3}\left(\sum_{\alpha : \pi_3(\alpha) \neq 0}\left\vert\hat{f}_u(\pih(\alpha))\right\vert\left\vert\hat{g}_v(\alpha)\right\vert^2\left(\frac{1}{2}\right)^{\#\alpha}\right),
\end{equation*}
meaning that
\begin{equation}
\frac{3\eps}{4}
\leq \sum_{\alpha: \pi_3(\alpha) \neq 0}\left\vert\hat{f}_u(\pih(\alpha))\right\vert\left\vert\hat{g}_v(\alpha)\right\vert^2\left(\frac{1}{2}\right)^{\#\alpha}.\label{eq:fourier-preprob}
\end{equation}
Parseval's equation tells us that $\sum_{\alpha \in \Z_3^L}\vert \hat{g}_v(\alpha) \vert^2 = 1$.  The function $\hat{g}_v$ therefore induces a probability distribution on the elements of $\Z_3^L$.  As a result, we can rewrite Equation~\eqref{eq:fourier-preprob} as
\begin{equation}
\frac{3\eps}{4}\leq \E_{\alpha \sim \hat{g}_v}\left[\left\vert\hat{f}_u(\pih(\alpha))\right\vert\left(\frac{1}{2}\right)^{\#\alpha}\bone[\pi_3(\alpha) \neq 0]\right].\label{eq:fourier-prob}
\end{equation}
As previously noted, $\vert \hat{f}_u(\pih(\alpha))\vert$ is less than 1 for all $\alpha$, so the expression in this expectation as never greater than 1.  We can thus conclude that
\begin{equation*}
\frac{3\eps}{8}\leq \Pr_{\alpha \sim \hat{g}_v}\underbrace{\left[\left\vert \hat{f}_u(\pih(\alpha))\right\vert \left(\frac{1}{2}\right)^{\#\alpha} \bone[\pi_3(\alpha) \neq 0]\geq \frac{3\eps}{8}\right]}_{\mathsf{GOOD}_\alpha},
\end{equation*}
as otherwise the expectation in Equation~\eqref{eq:fourier-prob} would be less than $3\eps/4$.  Call the event in the probability $\mathsf{GOOD}_\alpha$.  When $\mathsf{GOOD}_\alpha$ occurs, the following happens:
\begin{itemize}
\item $\vert \hat{f}_u(\pih(\alpha))\vert^2 \geq 9\eps^2/64$.
\item $\#\alpha \leq \log_2(8/3\eps)$.
\item $\pi_3(\alpha) \neq 0$.  As a result, $\#\alpha > 0$.
\end{itemize}
This suggests the following randomized decoding procedure for each $u \in U$: pick an element $\beta \in \Z_3^K$ with probability $\vert \hat{f}_u(\beta)\vert^2$ and choose one of its nonzero coordinates uniformly at random.  Similarly, for each $v \in V$, pick an element $\alpha \in \Z_3^L$ with probability $\vert \hat{g}_v(\alpha)\vert^2$ and choose one of its nonzero coordinates uniformly at random.  In both cases, nonzero coordinates are guaranteed to exist because all the $f_u$'s and $g_v$'s are folded.

Now we analyze how well this decoding scheme performs for the edges $e = \{u, v\} \in E'$ (we may assume the other edges are unsatisfied).  Suppose that when the elements of $\Z_3^K$ and $\Z_3^L$ were randomly chosen, $g_v$'s set $\alpha$ was in $\mathsf{Good}_\alpha$, and $f_u$'s set $\beta$ equals $\pih(\alpha)$.  Then, as $\#\alpha \leq \log_2(8/3\eps)$, and each label in $\pih(\alpha)$ has at least one label in $\alpha$ which maps to it, the probability that matching labels are drawn is at least $1/\log_2(8/3\eps)$.  Next, the probability that such an $\alpha$ and $\beta$ are drawn is
\begin{equation*}
\sum_{\alpha \in \mathsf{GOOD}}\vert \hat{f}_u(\pih(\alpha)) \vert^2\vert \hat{g}_v(\alpha)\vert^2
\geq \frac{9\eps^2}{64}\sum_{\alpha \in \mathsf{GOOD}}\vert \hat{g}_v(\alpha)\vert^2
\geq \frac{9\eps^2}{64}\frac{3\eps}{8} = \frac{27\eps^3}{512}.
\end{equation*}
Combining these, the probability that this edge is satisfied is at least $27\eps^3/512\log_2(8/3\eps)$.  Thus, the decoding scheme satisfies at least
\begin{equation*}
\frac{27\eps^3}{512\log_2(8/3\eps)}\cdot \frac{\vert E'\vert}{\vert E\vert} \geq \frac{27\eps^4}{1024\log_2(8/3\eps)}
\end{equation*}
fraction of the Label Cover edges in expectation.  By the probabilistic method, an assignment to the Label Cover instance must therefore exist which satisfies at least this fraction of the edges.

We now apply Theorem~\ref{thm:mosraz}, setting the soundness value in that theorem equal to $O(\eps^5)$, which concludes the proof.  
\end{proof}

\appendix

\section{Proof of $\E[g(\y)g(\z)g(\w)]$}\label{sec:eggg}

\begin{lemma}[Lemma~\ref{lem:eggg} restated]
$-\Re\E[g(\y)g(\z)g(\w)] \leq \frac{1}{2} - \frac{3}{2}\vert\hat{g}(0)\vert^2$.
\end{lemma}
Given $\alpha, \beta, \gamma \in \Z_3^L$, define the predicate $\psi(\alpha, \beta, \gamma)$ to be true whenever $\vert \alpha[i] \vert + \vert \beta[i] \vert + \vert \gamma[i] \vert \equiv 0$ for all $i$.  In addition, define the function $\Phi(\cdot, \cdot, \cdot)$ as
\begin{equation*}
\Phi(\alpha, \beta, \gamma) = \prod_{i \in [K]} \prod_{\pi(j) = i} \left( 1 - \frac{1}{2}(\#\rho_j + \#\sigma_j + \#\tau_j)\right)
\end{equation*}
where $\rho = \alpha + \beta$, $\sigma = \beta+\gamma$, and $\tau = \alpha + \gamma$.  We will begin by deriving the following expansion for the expectation:
\begin{lemma}
Let $g_1, g_2, g_3: \Z_3^L \rightarrow \Z_3$.  Then
\begin{equation*}
\E[g_1(\y)g_2(\z)g_3(\w)] = \sum_{\psi(\alpha, \beta, \gamma)}\hat{g}_1(\alpha)\hat{g}_2(\beta)\hat{g}_3(\gamma) \Phi(\alpha, \beta, \gamma).
\end{equation*}
\end{lemma}

\begin{proof}
Begin by expanding out $\E[g_1(\y)g_2(\z)g_3(\w)]$:
\begin{equation}
\E[g_1(\y)g_2(\z)g_3(\w)]= \sum_{ \alpha, \beta, \gamma \in \Z_3^L}
\hat{g}_1(\alpha)\hat{g}_2(\beta)\hat{g}_3(\gamma)\E[\chi_\alpha(\y)\chi_{\beta}(\z)\chi_\gamma(\w)].\label{eq:fourierexpanded-ggg}
\end{equation}
We focus on the products of the Fourier characters:
\begin{equation}
\E[\chi_\alpha(\y)\chi_\beta(\z)\chi_\gamma(\w)]
=  \prod_{i \in [K]}\E[\chi_{\alpha[i]}(\y[i]) \chi_{\beta[i]}(\z[i])\chi_{\gamma[i]}(\w[i])]\label{eq:fcharacters-ggg}
\end{equation}
We can attend to each block separately:
\begin{equation}
\E[\chi_{\alpha[i]}(\y[i]) \chi_{\beta[i]}(\z[i])\chi_{\gamma[i]}(\w[i])]
= \E_{\x}\left[\prod_{j:\pi(j) = i}\underbrace{\E_{\y, \z, \w}\left[\omega^{\alpha_j \y_j+  \beta_j \z_j + \gamma_j \w_j} \mid \x_i = a\right]}_{(*)}\right].\label{eq:blockcharacters-ggg}
\end{equation}
To analyze the expectation $(*)$, note that conditioned on $\x_i = a$, the distribution on the values for $(\y_j, \z_j)$ is uniform on the six possibilities $(a, a+1, a+1)$, $(a+1, a, a+1)$, $(a+1, a+1, a)$, $(a, a+2, a+2)$, $(a+2, a, a+2)$, $(a+2, a+2, a)$.  Then the expectation $(*)$ is equal to
\begin{equation*}
\frac{1}{6}\left(\omega^{a(\alpha_j + \beta_j + \gamma_j)}\left(
\omega^{\alpha_j + \beta_j} + \omega^{2(\alpha_j + \beta_j)}
+ \omega^{\alpha_j + \gamma_j} + \omega^{2(\alpha_j + \gamma_j)}
+ \omega^{\beta_j + \gamma_j} + \omega^{2(\beta_j + \gamma_j)}\right)\right).
\end{equation*}
Note that $\omega^{\alpha_j + \beta_j} + \omega^{2(\alpha_j + \beta_j)} = \omega^{\rho_j} + \omega^{2\rho_j} = 2 - 3\#\rho_j$.  Thus, the previous equation is equal to
\begin{equation*}
\frac{1}{6}\left(\omega^{a(\alpha_j + \beta_j + \gamma_j)}\left(
6 - 3 \#\rho_j - 3\#\sigma_j - 3\#\tau_j\right)\right)
= \omega^{a(\alpha_j + \beta_j + \gamma_j)}\left(
1 - \frac{1}{2}( \#\rho_j +\#\sigma_j +\#\tau_j)\right)
\end{equation*}
Substituting this into Equation~\eqref{eq:blockcharacters-ggg} yields
\begin{equation*}
\E[\chi_{\alpha[i]}(\y[i]) \chi_{\beta[i]}(\z[i])\chi_{\gamma[i]}(\w[i])]
= \E_{\x}\left[\omega^{a(\vert \alpha[i] \vert + \vert \beta[i] \vert + \vert \gamma[i] \vert)}\prod_{j:\pi(j) = i} \left(1 - \frac{1}{2}( \#\rho_j +\#\sigma_j +\#\tau_j)\right)\right],
\end{equation*}
which is zero unless $\vert \alpha[i] \vert + \vert \beta[i] \vert + \vert \gamma[i] \vert \equiv 0$.  Thus, the only Fourier coefficients $\alpha$, $\beta$, and $\gamma$ which contribute to Equation~\eqref{eq:fourierexpanded-ggg} are those which satisfy $\psi(\alpha, \beta, \gamma)$.  Furthermore, for any such $\alpha$, $\beta$, and $\gamma$, this equation is just
\begin{equation*}
\prod_{j:\pi(j) = i} \left(1 - \frac{1}{2}( \#\rho_j +\#\sigma_j +\#\tau_j)\right),
\end{equation*}
and so Equation~\eqref{eq:fcharacters-ggg} is equal to $\Phi(\alpha, \beta, \gamma)$.  This concludes the lemma.
\end{proof}

Define $p = \Pr[g(\y) = 0]$, $q = \Pr[g(\y) = \omega]$, and $r = \Pr[g(\y) = \omega^2]$.  For $a \in \roots$, $\bone_a(y)$ is the indicator of the event $g(y) = a$.  In the Fourier transform calculations that follow, we may write $\hat{f}^3$ for some function $f$.  This is always shorthand for $\hat{f}(\alpha) \hat{f}(\beta) \hat{f}(\gamma)$.  Similarly, given three functions $f_1$, $f_2$, and $f_3$, $\hat{f}_1\hat{f}_2\hat{f}_3$ is always shorthand for $\hat{f}_1(\alpha)\hat{f}_2(\beta)\hat{f}_3(\gamma)$.

We'll start by computing the value of $\vert \hat{g}(\vec{0}) \vert^2$.
\begin{proposition}\label{prop:hat-g}
$\vert \hat{g}(\vec{0}) \vert^2 = p^3 + q^3 + r^3 - 3\cdot pqr.$
\end{proposition}
\begin{proof}
Direct calculation shows that $\vert \hat{g}(\vec{0}) \vert^2 = p^2 + q^2 + r^2 - (pq + pr + rq)$.  Then because $p+q+r = 1$,
\begin{align*}
p^2 + q^2 + r^2 - (pq + pr + rq)
& = p^2 + q^2 + r^2 - (pq + pr + rq)\cdot(p+q+r)\\
& = p^2 + q^2 + r^2 - p^2 q - p^2 r - q^2 p - q^2 r - r^2 p - r^2 q - 3\cdot pqr\\
& = p^3 + q^3 + r^3 - 3 \cdot pqr,
\end{align*}
where the last step replaces $p^2 - p^2 q - p^2 r$ with $p^3$, using $1 - q - r = p$, and performs similar replacements for $q$ and $r$.
\end{proof}

Before proving Lemma~\ref{lem:eggg}, we'll need the following lemma.
\begin{lemma}\label{lem:threeones}
$\sum_{a \in \roots} \E[\bone_a(\x) \bone_a(\z) \bone_a(\w)] 
= 3 \cdot \E[\bone_1(\y)\bone_{\omega}(\z)\bone_{\omega^2}(\w)] + \vert \hat{g}(\vec{0}) \vert^2.$
\end{lemma}
\begin{proof}
The LHS is equal to
\begin{equation}
\Re \sum_{\psi(\alpha, \beta, \gamma)}(\hat{\bone}_1^3 + \hat{\bone}_{\omega}^3 + \hat{\bone}_{\omega^2}^3)\cdot \Phi(\alpha, \beta, \gamma) \label{eq:sum-fourier}
\end{equation}
Because $\bone_1 + \bone_\omega + \bone_{\omega^2} = 1$, the Fourier coefficients of $\bone_{\omega^2}$ may be rewritten as follows:
\begin{equation*}
\hat{\bone}_{\omega^2} = \bone[\alpha = 0] -\hat{\bone}_1 - \hat{\bone}_{\omega}.
\end{equation*}
When $\alpha \neq \vec{0}$, the corresponding term in Equation~\eqref{eq:sum-fourier} equals
\begin{align*}
\hat{\bone}_1^3 + \hat{\bone}_{\omega}^3 + \hat{\bone}_{\omega^2}^3
 =&~\hat{\bone}_1^3 + \hat{\bone}_{\omega}^3 + (-\hat{\bone}_{1}-\hat{\bone}_{\omega})^3\\
 =&~\hat{\bone}_1(\beta) \hat{\bone}_{\omega}(\gamma)\cdot(-\hat{\bone}_1(\alpha) - \hat{\bone}_{\omega}(\alpha))
  + \hat{\bone}_1(\gamma) \hat{\bone}_{\omega}(\alpha)\cdot(-\hat{\bone}_1(\beta) - \hat{\bone}_{\omega}(\beta))\\
  &+ \hat{\bone}_1(\alpha) \hat{\bone}_{\omega}(\beta)\cdot(-\hat{\bone}_1(\gamma) - \hat{\bone}_{\omega}(\gamma))\\
=&~\hat{\bone}_{\omega^2}\hat{\bone}_1\hat{\bone}_{\omega}
	+ \hat{\bone}_{\omega}\hat{\bone}_{\omega^2}\hat{\bone}_1
	+ \hat{\bone}_1\hat{\bone}_\omega\hat{\bone}_{\omega^2}.
\end{align*}
For the $\alpha = \vec{0}$ case, 
\begin{equation*}
p^3 + q^3 + r^3 = 3\cdot pqr + p^3 + q^3 + r^3 - 3\cdot pqr
= 3\cdot pqr + \vert \hat{g}(\vec{0}) \vert^2,
\end{equation*}
where the last step uses Proposition~\ref{prop:hat-g}.  Substituting these into Equation~\eqref{eq:sum-fourier} yields
\begin{equation*}
\eqref{eq:sum-fourier} = 
\vert \hat{g}(\vec{0})\vert^2 +
\Re \sum_{\psi(\alpha, \beta, \gamma)}(
	\hat{\bone}_1\hat{\bone}_\omega\hat{\bone}_{\omega^2}
	+ \hat{\bone}_{\omega^2}\hat{\bone}_1\hat{\bone}_{\omega}
	+ \hat{\bone}_{\omega}\hat{\bone}_{\omega^2}\hat{\bone}_1
	) \cdot \Phi(\alpha, \beta, \gamma),
\end{equation*}
which equal the RHS of the lemma.
\end{proof}

We now proceed to the proof of Lemma~\ref{lem:eggg}
\begin{proof}[Proof of Lemma~\ref{lem:eggg}]
Rewrite $g$ in terms of its indicator functions, i.e. $g = \bone_{1} + \omega \cdot \bone_{\omega} + \omega^2 \cdot \bone_{\omega^2}$.  Then
\begin{align*}
\Re\E[g(\y)g(\z)g(\w)]
 &= \Re \sum_{\psi(\alpha, \beta, \gamma)} \hat{g}^3\cdot \Phi(\alpha, \beta, \gamma) \\
 &= \Re \sum_{\psi(\alpha, \beta, \gamma)} (\hat{\bone}_{1} + \hat{\bone}_{\omega} + \hat{\bone}_{\omega^2})^3\cdot \Phi(\alpha, \beta, \gamma) \\
 &= \Re \sum_{\psi(\alpha, \beta, \gamma)} \sum_{a, b, c \in \roots} a b c\cdot \hat{\bone}_a \hat{\bone}_b \hat{\bone}_c \cdot \Phi(\alpha, \beta, \gamma) \\
 &= \sum_{\psi(\alpha, \beta, \gamma)} \left(\sum_{abc = 1}  \hat{\bone}_a \hat{\bone}_b \hat{\bone}_c
 	- \frac{1}{2}\sum_{abc \neq 1} \hat{\bone}_a \hat{\bone}_b \hat{\bone}_c\right) \cdot \Phi(\alpha, \beta, \gamma) \\
&=  \sum_{abc = 1} \E[\bone_a(\x) \bone_b(\z) \bone_c(\w)] - \frac{1}{2} \sum_{abc \neq 1} \E[\bone_a(\x) \bone_b(\z) \bone_c(\w)]\\
&=  \sum_{abc = 1} \E[\bone_a(\x) \bone_b(\z) \bone_c(\w)] - \frac{1}{2} \left(1-\sum_{abc = 1} \E[\bone_a(\x) \bone_b(\z) \bone_c(\w)]\right)\\
&=  \frac{3}{2}\cdot \sum_{abc = 1} \E[\bone_a(\x) \bone_b(\z) \bone_c(\w)] - \frac{1}{2},
\end{align*}
where the second-to-last step uses the fact that $\sum_{a, b, c \in \roots} \E[\bone_a(\x)\bone_b(\y)\bone_c(\z)] = 1$.  This concludes the proof, as
\begin{align*}
\sum_{abc = 1} \E[\bone_a(\x) \bone_b(\z) \bone_c(\w)]
 &= \sum_{a \in \roots} \E[\bone_a(\x) \bone_a(\z) \bone_a(\w)]+ 6\cdot  \E[\bone_{1}(\x) \bone_{\omega}(\z) \bone_{\omega^2}(\w)]\\
 & = \vert \hat{g}(\vec{0}) \vert^2 + 9\cdot  \E[\bone_{1}(\x) \bone_{\omega}(\z) \bone_{\omega^2}(\w)] \tag{By Lemma~\ref{lem:threeones}.}\\
 & \geq \vert \hat{g}(\vec{0}) \vert^2,
\end{align*}
using the fact that $\E[\bone_{1}(\x) \bone_{\omega}(\z) \bone_{\omega^2}(\w)] \geq 0$.
\end{proof}

\bibliographystyle{alpha}
\bibliography{odonnell-bib}

\end{document}